\documentclass[journal]{IEEEtran}

\usepackage{cite}
\usepackage{graphicx}
\usepackage{color}
\usepackage{amsmath, bm}

\DeclareMathOperator{\diag}{diag}
\newcommand{\norm}[1]{\left\lVert#1\right\rVert}
\newcommand{\abs}[1]{\left\lvert#1\right\rvert}

\usepackage{amsthm}
\theoremstyle{remark}
\newtheorem*{remark}{Remark}
\theoremstyle{definition}
\newtheorem{definition}{Definition}
\newtheorem{theorem}{Theorem}

\newtheorem{proposition}[theorem]{Proposition}

\begin{document}

\title{On Spatial Multiplexing Using Reconfigurable Intelligent Surfaces}

\author{Mohamed~A.~ElMossallamy,~\IEEEmembership{Student Member,~IEEE,}
        Hongliang~Zhang,~\IEEEmembership{Member,~IEEE,}
        Radwa~Sultan,~\IEEEmembership{Member,~IEEE,}
        Karim~G.~Seddik,~\IEEEmembership{Senior~Member,~IEEE,}
        Lingyang~Song,~\IEEEmembership{Fellow,~IEEE,}
        Geoffrey~Ye~Li,~\IEEEmembership{Fellow,~IEEE} and 
        Zhu~Han,~\IEEEmembership{Fellow,~IEEE} \vspace{-5mm}
        \thanks{Mohamed A. ElMossallamy and Zhu Han are with the Electrical and Computer Engineering Department, University of Houston, TX, USA. 
                Zhu Han is also with the Department of Computer Science and Engineering, Kyung Hee University, Seoul, South Korea (emails: m.ali@ieee.org, zhan2@uh.edu).}
        \thanks{Hongliang Zhang is with the Electrical Engineering Department, Princeton University, Princeton, New Jersey, USA (email: hongliang.zhang92@gmail.com).}
        \thanks{Radwa Sultan is with the Electrical and Computer Engineering Department, Manhattan College, Riverdale, NY, USA (email: rsultan02@manhattan.edu).}
        \thanks{Karim G. Seddik is with the Electronics and Communications Engineering Department, American University in Cairo, New Cairo, Egypt (email: kseddik@aucegypt.edu).}
        \thanks{Lingyang Song is with Department of Electronics, Peking University, Beijing, China (email: lingyang.song@pku.edu.cn).}
        \thanks{Geoffrey Ye Li is with the School of Electrical and Computer Engineering, Georgia Institute of Technology, Atlanta, GA, USA (email: liye@ece.gatech.edu).}
}

\maketitle

\begin{abstract}
 We consider an uplink multi-user scenario and investigate the use of reconfigurable intelligent surfaces (RIS) to optimize spatial multiplexing performance when a linear receiver is used. We study two different formulations of the problem, namely maximizing the effective rank and maximizing the minimum singular value of the RIS-augmented channel. We employ gradient-based optimization to solve the two problems and compare the solutions in terms of the sum-rate achievable when a linear receiver is used. Our results show that the proposed criteria can be used to optimize the RIS to obtain effective channels with favorable properties and drastically improve performance even when the propagation through the RIS contributes a small fraction of the received power.
\end{abstract}

\IEEEpeerreviewmaketitle

\section{Introduction}
{Performance of wireless communication systems} is ultimately dictated by the wireless channel {state}, which is the product of electromagnetic wave propagation in the environment, and outside the control of the system designer. Transceivers merely track the channel state, and then adapt their modulation and coding {to better utilize it for a given coherence interval.} Reconfigurable intelligent surfaces (RIS) present a paradigm shift in this aspect \cite{basar_wireless_2019, basar_reconfigurable_2020, elmossallamy_reconfigurable_2020, huang_holographic_2020}, as it allows the system designer to also alter the wireless channel realization to increase capacity or impose favorable structures to facilitate simpler communication techniques.

RISs have been used to alter the channel to achieve various objectives in many scenarios \cite{huang_reconfigurable_2019, wu_joint_TWC_2019, nadeem_intelligent_2019, guo_weighted_2019, huang_reconfigurable_2020, yang_irs_meets_ofdm_2019, taha_enabling_2019, hu_reconfigurable_2019, di_hybrid_2020, nadeem_asymptotic_2020}. Downlink multi-user scenarios has been studied in \cite{huang_reconfigurable_2019, wu_joint_TWC_2019, nadeem_intelligent_2019, guo_weighted_2019, huang_reconfigurable_2020}, where the use of an RIS resulted in impressive gains both in terms of energy efficiency \cite{huang_reconfigurable_2019, wu_joint_TWC_2019} and achievable rates \cite{guo_weighted_2019, huang_reconfigurable_2020}. In \cite{wu_joint_TWC_2019, nadeem_intelligent_2019}, it has been demonstrated that an RIS-assisted MIMO system can achieve the same rate performance as massive MIMO systems while significantly reducing the number of required radio frequency (RF) chains. Gains in the rate performance have also been reported for wide-band OFDM systems \cite{yang_irs_meets_ofdm_2019, taha_enabling_2019}. Furthermore, it has been shown in \cite{hu_reconfigurable_2019}, using both simulations and experimental prototypes, that an RIS can be exploited to greatly improve RF sensing for human posture recognition.

In this letter, we investigate the use of an RIS to improve the performance in a multi-user uplink scenario \emph{when a linear receiver is used}. Although suboptimal in general, linear receivers are becoming increasingly attractive as larger and larger constellations are being incorporated in wireless standards. For example, the current generation of wireless local area network (WLAN) modems supports 1024-QAM constellations with 4096-QAM constellations being proposed for the next generation \cite{schelstraete_4KQAMfeasbility_2019}, which renders maximum-likelihood (ML) receivers prohibitive. We present two optimization criteria for the RIS phase shifts to improve the post-processing signal-to-interference-plus-noise ratio (SINR) of a linear receiver, and ensure fairness between users. Different from alternating optimization techniques in the literature, e.g., \cite{wu_joint_TWC_2019}, we propose to optimize features of the effective channel matrix directly. In particular, maximizing the so-called effective rank \cite{roy_effective_2007} and maximizing the minimum singular value of the RIS-augmented channel. We derive the gradient of the two criteria with respect to the RIS phase shifts and employ gradient-based optimization techniques to obtain good RIS configurations. We compare the configurations obtained in terms of rate achievable by a linear receiver and examine the characteristics of the resulting RIS-augmented channels. Finally, we investigate the effects of changing system parameters such as the number of RIS elements and the fraction of power received through the RIS. Our results show that by optimizing the RIS phase shifts using the proposed criteria, the rate achievable by a linear receiver is drastically increased even when the propagation through the RIS contributes a small fraction of the received power.

\section{System Model}
\begin{figure}[t] 
  \includegraphics[height=6cm,clip]{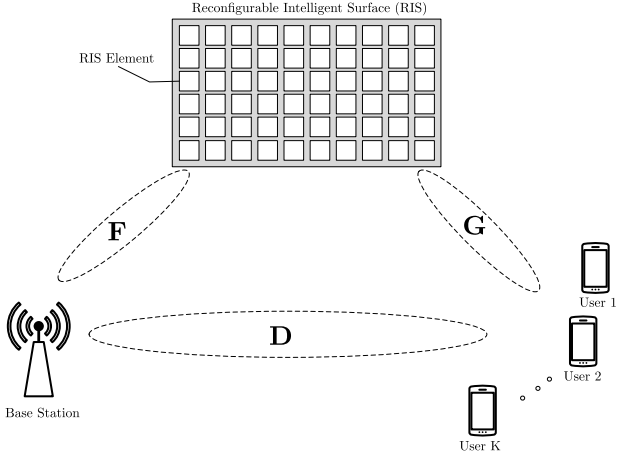} 
  \centering
  \caption{System Model.} 
  \label{fig:sysmodel}
\end{figure}
We consider an uplink multi-user scenario with one base-station having $M$ antennas and $K$ single-antenna users, where $M \ge K$. The link is assisted by a single $L$-element reflectarray-based RIS \cite{elmossallamy_reconfigurable_2020}. The baseband received signal at the base station can be written as
\begin{equation}\label{eq:rcvdsigVec}
\begin{split}
    \mathbf{y} = \sum_{k = 1}^K \mathbf{h}_k x_k + \mathbf{n},
\end{split}
\end{equation}
where $\mathbf{h}_k$ is the $M \times 1$ channel vector from the $k$-th user to the base-station including the effects of the RIS, $x_k$ is the transmitted scalar symbol by the $k$-th user, and $\mathbf{n}$ is the $M \times 1$ noise vector at the base-station whose elements are distributed as i.i.d. $\mathcal{CN}\left(0,\sigma^2_n\right)$. The received signal can be written in matrix form as 
\begin{equation}\label{eq:rcvdlambdat}
\begin{split}
    \mathbf{y} = \mathbf{H}_{\textup{eff}} \mathbf{x} + \mathbf{n},
\end{split}
\end{equation}
where $\mathbf{H}_{\textup{eff}}=\left[ \mathbf{h}_1 \mathbf{h}_2 \dots \mathbf{h}_K \right]$ is the overall effective channel including the effects of the RIS and $\mathbf{x}=\left[ x_1, x_2, \dots, x_K \right]$.

By changing the RIS configuration, the effective channel between the base station and the users, $\mathbf{H}_{\textup{eff}}$ can be altered. Assuming a reflectarray-based RIS \cite{elmossallamy_reconfigurable_2020}, we can write the RIS-assisted channel as
\begin{equation}\label{eq:risChannel_dbc}
    \begin{split}
        \mathbf{H}_{\textup{eff}} = \sqrt{1-\alpha} ~\mathbf{D} + \frac{\sqrt{\alpha}}{\sqrt{L}} \mathbf{F} \mathbf{Q} \mathbf{G}, 
    \end{split}
\end{equation}
where $\mathbf{D}$ denote the $M \times K$ direct, i.e., {not controllable by the RIS}, channel between the users and the base-station, $\alpha \in \left[0, 1\right]$ denote the fraction of power received through the RIS\footnote{The value of the parameter $\alpha$ will depend on various parameters such as the deployment scenario, and the physical size of the RIS. Note that the normalization by $\frac{1}{\sqrt{L}}$ is necessary for $\alpha$ to represent the fraction of power received through the RIS.}, $\mathbf{F}$ denotes the $M \times L$ channel between the base station and the RIS, $\mathbf{G}$ denotes the $L \times K$ channel between the RIS and the $K$ users, and $\mathbf{Q}$ denote the RIS controllable interaction matrix. We assume the elements of all channels' matrices are distributed as i.i.d. $\mathcal{CN}\left(0, 1\right)$, i.e., Rayleigh fading\footnote{It is worth mentioning that obtaining channel knowledge at the passive RIS array is a challenging problem, but outside the scope of the current work.}.

Assuming no coupling between the RIS antenna elements, the interaction matrix can be written as
\begin{equation}\label{eq:interactionMat_DBC}
\begin{split}
\mathbf{Q} = \diag \left( \psi_1, \psi_2 , \dots, \psi_L \right),
\end{split}
\end{equation}
where $\psi_i$ represents the reflection coefficient of the $i$-th element. We assume $\abs{\psi_i} = 1~\forall i$ since the RIS does not possess any amplification capabilities and can only phase shift the incident signals.

\section{Optimization Criteria for RIS-assisted Spatial Multiplexing}
In this section, we present two different criteria to choose the RIS configuration to optimize spatial multiplexing performance when a linear receiver is used. In general, when a linear equalizer, $\mathbf{W}$, is used, the resulting post-equalization SINR of the $k$-th user is given by
\begin{equation}\label{eq:sinr}
\gamma_k = \frac{\abs{\mathbf{w}^*_k \mathbf{h}_k}^2}{\sum\limits_{j \neq k} \abs{\mathbf{w}^*_k \mathbf{h}_j}^2 + K \sigma_n^2 \norm{\mathbf{w}_k}^2},
\end{equation}
where $\mathbf{w}^*_k$ is the $k$-th row of $\mathbf{W}$. The optimal \emph{linear} receiver that maximizes this ratio is the minimum mean square error (MMSE) receiver given by 
\begin{equation}\label{eq:mmse_rx}
\mathbf{W}_{\textrm{MMSE}} = \left( \mathbf{H}_{\textup{eff}}^* \mathbf{H}_{\textup{eff}}+ \sigma_n^2 \mathbf{I} \right)^{-1} \mathbf{H}_{\textup{eff}}^*,
\end{equation}
where $\left( \cdot \right)^*$ denotes the Hermitian transpose and $\mathbf{I}$ denotes the identity matrix of suitable dimensions. The MMSE receiver strikes a balance between minimizing the interference terms, i.e., $\sum_{j \neq k} \abs{\mathbf{w}^*_k \mathbf{h}_j}^2$, and maximizing the signal term $\abs{\mathbf{w}^*_k \mathbf{h}_k}^2$.

Aside from choosing the receiver structure, now with the advent of RIS, we also have the capability to change the channel realization. A natural question to ask is \emph{for which channel realization is the linear MMSE receiver equivalent to the optimal ML receiver?} the answer is: when the channel is orthogonal, i.e., $\mathbf{H}_{\textup{eff}}^* \mathbf{H}_{\textup{eff}}$ is diagonal, and no stream carries any information about other streams. Also note that in this case power control simplifies considerably, as it becomes optimal for each user to just transmit at full power. Furthermore, in this case, the interference terms $\sum_{j \neq k} \abs{\mathbf{w}^*_k \mathbf{h}_j}^2$ vanish and the MMSE receiver reduces to the matched filter  $\mathbf{W}_{\textup{MF}} = \mathbf{H}_{\textup{eff}}^*$.\footnote{The matched filter receiver could be practically attractive in envisioned massive connectivity scenario where the high dimensionality of the channel matrix makes its inversion prohibitive.} The problem of finding the phase shifts that orthogonalize the effective channel, $\mathbf{H}_{\textup{eff}}$, can be written in the form of the feasibility problem 
\begin{equation}
\begin{aligned}
& \text{find}
& & \boldsymbol{\psi} \\
& \text{subject to}
& & \mathbf{h}^*_i \mathbf{h}_j  = 0, \; i \neq j,
\end{aligned}
\end{equation}
where $\boldsymbol{\psi} = \left[ \psi_1, \psi_2, \dots, \psi_L \right]$. This problem is equivalent to solving a system of multi-variate quadratics, which is known to be NP-complete \cite{garey_computers_1979}. Nevertheless, we present two optimization criteria that could be used to approximately orthogonalize the effective channel.\footnote{{Note that this formulation is motivated by simplicity of detection, power control and fairness between users but is inherently sub-optimal from a sum-rate stand-point.}}

\subsection{Effective Rank Criterion (ER-C)}
\begin{definition}
The effective rank \cite{roy_effective_2007,hougne_optimally_2019} of a complex-valued $M \times K$ matrix $\mathbf{X}$ is given by
\begin{equation}\label{eq:effectiveRank}
\xi\left(\mathbf{X}\right) = \exp{\left[- \sum_{i} \frac{\lambda_i}{\norm{\mathbf{X}}_{*}} \ln{\left(\frac{\lambda_i}{\norm{\mathbf{X}}_{*}}\right)}\right]},
\end{equation}
where $\lambda_i$ is the $i$-th largest singular value of $\mathbf{X}$ and $\norm{\mathbf{X}}_{*}=\sum\limits_{i} \lambda_i$ is the nuclear norm of $\mathbf{X}$. The term in the exponent is known as the \emph{spectral entropy} and the density defined by $\left\lbrace\frac{\lambda_i}{\norm{\mathbf{X}}_{*}}\right\rbrace_{\forall i}$ is known as the \emph{spectral density}.
\end{definition} 
\noindent 

The effective rank takes values in the range $\left[ 0, \min{\left( M, K \right)} \right]$, where the higher {its} value the more equal the singular values, and thus the more orthogonal the matrix columns are. A matrix with maximum effective rank, i.e., $\min{\left( M, K \right)}$, has completely orthogonal columns and all its singular values are equal. Hence, a proper optimization objective to achieve our goal is to maximize the effective rank of the effective channel matrix, i.e.,
\begin{equation}\label{prob:effective_rank}
\begin{aligned}
\underset{\psi_{\ell}}{\text{maximize}} \quad & \xi \left( \mathbf{H}_{\textup{eff}} \right)\\
\text{subject to} \quad & {\abs{\psi_{\ell}} = 1}~\forall {\ell} = 1, 2, \dots, L.\\
\end{aligned}
\end{equation}
Note that the effective rank is not a function of the singular values directly but their normalized form $\frac{\lambda_i}{\norm{\mathbf{X}}_{*}}$. 

We can write the interaction matrix directly in terms of the phase shifts \cite{huang_reconfigurable_2019} as $\mathbf{Q} = \diag \left(  \mathrm{e}^{i \theta_1}, \mathrm{e}^{i \theta_2}, \dots, \mathrm{e}^{i \theta_L} \right)$ to reformulate~\eqref{prob:effective_rank} as the unconstrained problem 
\begin{equation}\label{prob:effective_rank_unc}
\begin{aligned}
\underset{\boldsymbol{\theta}}{\text{maximize}} \quad & \xi \left( \mathbf{H}_{\textup{eff}} \right)
\end{aligned}
\end{equation}
where $\boldsymbol{\theta} = \left[ \theta_1, \theta_2, \dots, \theta_L \right]$. Although {not convex}, this problem can be efficiently solved using gradient-based techniques to {find a locally optimal solution.} To facilitate gradient-based optimization, we derive the gradient of the effective rank of an RIS-augmented channel with respect to the phase shifts. 

Using the chain rule, we can write the $\ell$-th element of the gradient as 
\begin{equation}\label{eq:effRankGrad1}
\frac{\partial \xi}{\partial \mathbf{\theta_{\ell}}} = 
\begin{array}{@{}c@{}}
\begin{bmatrix}\frac{\partial \lambda_1}{\partial \mathbf{\theta_{\ell}}} & \frac{\partial \lambda_2}{\partial \mathbf{\theta_{\ell}}} & \cdots &\frac{\partial \lambda_K}{\partial \mathbf{\theta_{\ell}}}\end{bmatrix}
\end{array}
\begin{bmatrix}
\frac{\partial \xi}{\partial \lambda_1}\\ \\
\frac{\partial \xi}{\partial \lambda_2}\\ \\
\vdots \\ \\
\frac{\partial \xi}{\partial \lambda_K}
\end{bmatrix}.
\end{equation}
The partial derivative of the effective rank with respect to the $k$-th singular value can be found by directly differentiating~\eqref{eq:effectiveRank} to be
\begin{equation}\label{eq:parReff_parLambda}
    \begin{split}
        \frac{\partial \xi}{\partial \lambda_k} = - \sum_{j=1}^{K} \frac{C_{j, k}}{\norm{\mathbf{H}_{\textup{eff}}}_{*}^2} \left( 1 + \ln{\frac{\lambda_j}{\norm{\mathbf{H}_{\textup{eff}}}_{*}}} \right) ~~ \xi\left(\mathbf{H}_{\textup{eff}}\right), 
    \end{split}
\end{equation}
where 
\begin{equation}\label{eq:parReff_parLambda2}
C_{j, k} = 
  \begin{cases}
    \sum\limits_{i \neq k} \lambda_i & \text{if } j = k,\\
    - \lambda_j                   & \text{if } j \neq k.
  \end{cases}
\end{equation}
\begin{proposition}
Let  $\mathbf{H}_{\textup{eff}}=\mathbf{U^{*} S V} $ denote the singular value decomposition (SVD) of $\mathbf{H}_{\textup{eff}}$. Then, the partial derivative of the $k$-th singular value with respect to the $\ell$-th phase shift can be found to be 
\begin{equation}\label{eq:risChannel_lambda}
        \frac{\partial \lambda_k}{\partial \mathbf{\theta_{\ell}}} =  \mathbf{u}^{*}_k ~\frac{\partial \mathbf{H}_{\textup{eff}}}{\partial \theta_{\ell}} ~ \mathbf{v}_k,
\end{equation}
where $\mathbf{v}_k$ and $\mathbf{u}_k$ are the $k$-th columns of $\mathbf{U}$ and $\mathbf{V}$, respectively.
\end{proposition}
\begin{proof}
The $k$-th singular value is given by
\begin{equation}\label{eq:lemma1_proof1}
        \lambda_k = \mathbf{u}^{*}_k ~ \mathbf{H}_{\textup{eff}} ~ \mathbf{v}_k,
\end{equation}
Then taking the partial derivative with respect to $\theta_{\ell}$ yields
\begin{equation}\label{eq:lemma1_proof2}
\begin{split}
            \frac{\partial \lambda_k}{\partial \theta_{\ell}} &= \Re{\lbrace \frac{\partial \mathbf{u}^*_k}{\partial \theta_l}  \mathbf{H}_{\textup{eff}} ~ \mathbf{v}_k + \mathbf{u}^*_k  \frac{\partial \mathbf{H}_{\textup{eff}}}{\partial \theta_{\ell}} ~ \mathbf{v}_k + \mathbf{u}^*_k  \mathbf{H}_{\textup{eff}} ~ \frac{\partial \mathbf{v}^*_k}{\partial \theta_l}\rbrace}, \\
            \frac{\partial \lambda_k}{\partial \theta_{\ell}} & \overset{(a)}{=} \Re{\lbrace  \lambda_k \frac{\partial \mathbf{u}^*_k}{\partial \theta_l}  \mathbf{u}_k  + \mathbf{u}^*_k  \frac{\partial \mathbf{H}_{\textup{eff}}}{\partial \theta_{\ell}} ~ \mathbf{v}_k + \lambda_k \mathbf{v}^*_k \frac{\partial \mathbf{v}_k}{\partial \theta_l}\rbrace} \\
            \frac{\partial \lambda_k}{\partial \theta_{\ell}} & \overset{(b)}{=} \Re{\lbrace \mathbf{u}^*_k  \frac{\partial \mathbf{H}_{\textup{eff}}}{\partial \theta_{\ell}} ~ \mathbf{v}_k\rbrace},
\end{split}
\end{equation}
where (a) follows from the fact that $\mathbf{H}_{\textup{eff}} \mathbf{v}_k = \lambda_k \mathbf{u}_k$ and $\mathbf{H}_{\textup{eff}}^* \mathbf{u}_k = \lambda_k \mathbf{v}_k$, while (b) follows from the fact that $\mathbf{u}_k^* \mathbf{u}_k = \mathbf{v}_k^* \mathbf{v}_k = 1 $.
\end{proof}
The $\left(m, k\right)$-th entry of effective channel matrix $\mathbf{H}_{\textup{eff}}$ can be written as
\begin{equation}\label{eq:risChannel_elem}
    \begin{split}
        \left[\mathbf{H}_{\textup{eff}}\right]_{m, k}  &=  \sqrt{1-\alpha} \left[\mathbf{D}\right]_{m, k} + \frac{\sqrt{\alpha}}{\sqrt{L}} \sum_{\ell=1}^{L} \mathrm{e}^{\mathrm{i}\theta_{\ell}} \left[\mathbf{F}\right]_{m, \ell} \left[\mathbf{G}\right]_{\ell, k},
    \end{split}
\end{equation}
whose partial derivative with respect to the $\ell$-th phase shift is given by
\begin{equation}\label{eq:risChannel_elemD}
    \begin{split}
        \left[ \frac{\partial \mathbf{H}_{\textup{eff}}}{\partial \theta_{\ell}}\right]_{m, k}  &=  \frac{\sqrt{\alpha}}{\sqrt{L}} \mathrm{e}^{\mathrm{i} \left( \theta_{\ell} + \frac{\pi}{2} \right)} \left[\mathbf{F}\right]_{m, \ell} \left[\mathbf{G}\right]_{\ell, k}. 
    \end{split}
\end{equation}
Hence, the partial derivative of the entire effective channel matrix with respect to the $\ell$-th phase shift can be written as
\begin{equation}\label{eq:risChannel_matD}
    \begin{split}
        \frac{\partial \mathbf{H}_{\textup{eff}}}{\partial \theta_{\ell}}  &= \frac{\sqrt{\alpha}}{\sqrt{L}} \mathrm{e}^{\mathrm{i} \left( \theta_{\ell} + \frac{\pi}{2} \right)} \left[\mathbf{F}\right]_{:, \ell} \otimes \left[\mathbf{G}\right]_{\ell, :}, 
    \end{split}
\end{equation}
where $\otimes$ denote the Kronecker product operation. Finally, by substituting~\eqref{eq:parReff_parLambda}~and~\eqref{eq:risChannel_elemD} back into~\eqref{eq:effRankGrad1}, we can compute the gradient.

\subsection{Max-Min Singular Value Criterion (MSV-C)}
The minimum singular value of the channel matrix is paramount in determining the performance of linear receivers \cite{heath_antenna_2001}. Both capacity and error rate metrics are directly related to the minimum singular value which motivates our interest in this optimization criterion. Our numerical results show that, for an adequate size of the RIS and path loss values, maximizing the minimum singular value leads to the effective channel being completely orthogonalized. Solutions obtained by maximizing the minimum singular value also tend to have higher beamforming gains which lead to much better performance as compared to solutions obtained by maximizing the effective rank. 

Following the same arguments as in the last section, we can formulate the optimization as an unconstrained maximization problem with respect to the phase shifts as 
\begin{equation}\label{prob:min_singular}
\begin{aligned}
\underset{\boldsymbol{\theta}}{\text{maximize}} \quad & \lambda_{min} \left( \mathbf{H}_{\textup{eff}} \right),
\end{aligned}
\end{equation}
and the elements of the gradient of the minimum singular value with respect to the RIS phase shifts were already derived in~\eqref{eq:risChannel_lambda}. {Note that~\eqref{prob:min_singular}, like~\eqref{prob:effective_rank_unc}, is non-convex.}

Equipped with closed-form expressions for the gradient, both \eqref{prob:effective_rank_unc} and \eqref{prob:min_singular} can be solved efficiently using any gradient-based optimization algorithm to find a locally optimal solution. In our experiments, we have found the effective rank problem \eqref{prob:effective_rank_unc}, to be especially amicable to gradient-based optimization. Even the simplest steepest-descent algorithms can used to obtain an optimal solution, i.e., $\xi = \min{\left( M, K \right)}$. While, the max-min singular value problem, \eqref{prob:min_singular}, was found to benefit from incorporating curvature information using Quasi-Newton methods. In all cases, our numerical results presented in the next section show that the proposed problems are easily solvable using efficient numerical techniques and result in drastic improvements in the performance.

\section{Numerical Results}

\begin{figure}[t] 
  \includegraphics[height=7cm,trim={0.5cm 0 0.5cm 0},clip]{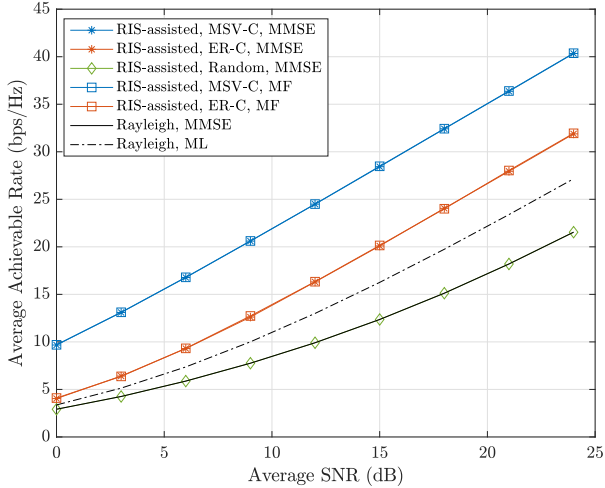}
  \centering
  \caption{Comparison of the achievable rate when the phase shifts are optimized according to MSV-C and ER-C. $\alpha = 0.5$ and $L=100$ for all curves. The achievable rate of the canonical Rayleigh channel is also shown for the ML and MMSE receivers as baselines.} 
  \label{fig:1}
\end{figure}

Equipped with the gradient expressions derived in the last section, we employ the Quasi-Newton Broyden-Fletcher-Goldfarb-Shanno (BFGS) algorithm with cubic interpolation line search \cite{nocedal_numerical_2006} to solve both \eqref{prob:effective_rank_unc} and \eqref{prob:min_singular}. Then, we compare both solutions in terms of the average achievable sum-rate attained. The achievable sum-rate with a linear receiver for a given channel realization can be computed from
\begin{equation}\label{eq:achievablerate}
R = \sum_{k = 1}^{K} \log_2{\left( 1 + \gamma_k \right)},
\end{equation}
where $\gamma_k$'s are as defined in~\eqref{eq:sinr}.
{In all simulated cases, we let $M=K=4$ and indicate the number of RIS elements, $L$, and the power fraction, $\alpha$ on the figures.}   Achievable rates are averaged over $10^6$ independent channel realizations. 

Fig.~\ref{fig:1} shows the attained average achievable rates for the RIS-assisted channels for both criteria and also for the case of a non-assisted canonical Rayleigh channel as a baseline. For all curves in the figure, $\alpha= 0.5$ and $L= 100$. From the figure, both optimization criteria lead to drastic improvements in the achievable rate over the canonical Rayleigh channel. Actually, the rate achieved with a linear MMSE receiver in the RIS-assisted channel\textemdash for both optimization criteria\textemdash exceeds even the ML receiver in a non-assisted channel. Moreover, the MF receiver performance is identical to the optimal MMSE receiver performance for both criteria which means that the effective RIS-assisted channel matrix is completely orthogonal in this case. \emph{This shows that the asymptotic orthogonality of massive MIMO can be achieved by a passive RIS without requiring a large number of active RF chains.} As expected, the MSV-C holds a significant advantage over the ER-C in terms of the achievable rate. Although both criteria lead to an orthogonal effective channel, the MSV-C incentivizes larger singular values, while the ER-C only cares about orthogonality, i.e., the singular values being equal. 

\begin{remark}
Note that the channel matrix is normalized such that the RIS assistance does not lead to increasing the average power at the receiver. All the achieved gains come from co-phasing at the RIS, i.e., beamforming, and the improved channel eigenstructure. To highlight this point,  Fig.~\ref{fig:1} also shows the achievable rate when the phase shifts are chosen randomly. In this case, all gains vanish and RIS-assisted channel holds no advantage over the canonical Rayleigh channel which underscores the importance of optimizing the phase shifts. {Our rationale for this normalization is to highlight the gains arising from the configurability of the RIS rather than the introduction of another large scatterer into the environment since the RIS will probably be installed into an already existing structure, e.g., building facade or wall.}
\end{remark}

\begin{figure}[t] 
  \includegraphics[height=7cm, trim={0.5cm 0 0.5cm 0},clip]{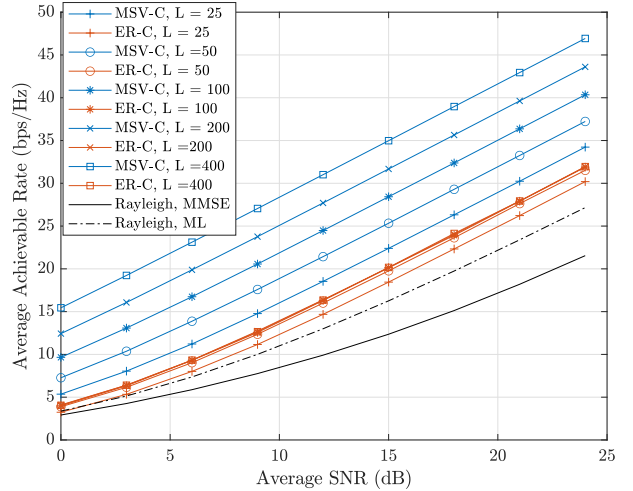} 
  
  \centering
  \caption{The effect of the number of RIS elements, $L$, on the achievable rate when the phase shifts are optimized according to MSV-C and ER-C. $\alpha=0.5$ for all curves. The achievable rate of the canonical Rayleigh channel is also shown for the ML and MMSE receivers as baselines.} 
  \label{fig:2}
\end{figure}

Now, we investigate the effects of the number of elements, $L$, and how much of the received power came through backscattering off the RIS, i.e., $\alpha$. Fig.~\ref{fig:2} shows the effect of the number of elements on the achievable rates for the RIS-assisted channels for both criteria and also for the case of a non-assisted canonical Rayleigh channel as a baseline. {For all curves in the figure, $\alpha= 0.5$.} From the figure, increasing the number of RIS elements leads to significant improvements in the achievable rate when the MSV-C is used to configure the phase shifts. However, when the ER-C is used, the effect is negligible. This stems from the fact the ER-C is channel gain agnostic and only cares about eliminating interference; hence, it does not leverage the increased beamforming capability that comes with more elements. It is worth mentioning that increasing the number of RIS elements entails increasing the physical size of the RIS since we are assuming the elements of $\mathbf{F}$ and $\mathbf{G}$ remain independent.

Finally, Fig.~\ref{fig:3} shows the effect of the power fraction, $\alpha$, on the achievable rates for the RIS-assisted channels for both criteria and also for the case of a non-assisted canonical Rayleigh channel as a baseline. {For all curves in the figure, $L = 100$.}  As expected, if more power is received through the RIS, it has a stronger influence on the effective channel and the achievable rate increases for both criteria.  Surprisingly, the RIS can have a potent effect on the effective channel even when propagation through it contributes little received power. Even for power fraction factors as low as $\frac{1}{32}$, the RIS can still influence the effective channel and drastically improve the achievable rate by up to 10 bits per channel use. 

\begin{figure}[t] 
  \includegraphics[height=7cm,trim={0.5cm 0 0.5cm 0},clip]{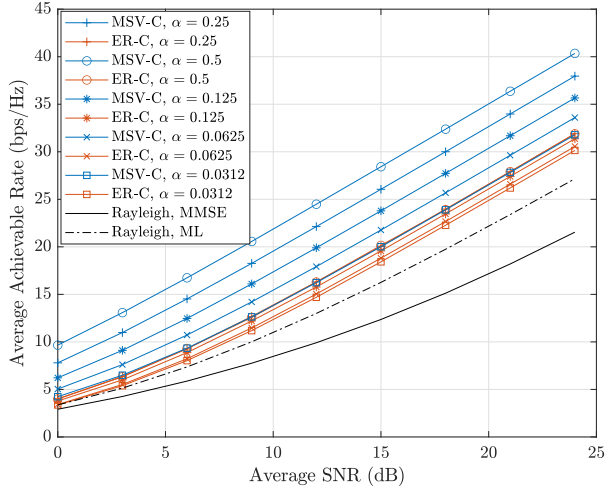} 
  
  \centering
  \caption{The effects of the power fraction, $\alpha$, on the achievable rate when the phase shifts are optimized according to MSV-C and ER-C. $L=100$ for all curves. The achievable rate of the canonical Rayleigh channel is also shown for the ML and MMSE receivers as baselines.} 
  \label{fig:3}
\end{figure}

\section{Conclusion}
We have presented two optimization criteria for the RIS phase shifts to improve the post-processing SINR of a linear receiver. In particular, maximizing the so-called effective rank and maximizing the minimum singular value of the RIS-augmented channel. We have derived the gradients of the two criteria with respect to the RIS phase shifts and employed gradient-based optimization techniques to obtain good RIS configurations. We have compared the configurations obtained in terms of rate achievable by a linear receiver and examined the characteristics of the resulting RIS-augmented channels. Finally, we investigated the effects of changing system parameters such as the number of RIS elements and the fraction of power received through the RIS. Our results showed that by optimizing the RIS phase shifts using the proposed criteria, the rate achievable by a linear receiver is drastically increased even when the propagation through the RIS contributes a small fraction of the received power.

\vspace{24pt}

\bibliographystyle{IEEEtran}
\bibliography{main}

\begin{thebibliography}{10}
\providecommand{\url}[1]{#1}
\csname url@samestyle\endcsname
\providecommand{\newblock}{\relax}
\providecommand{\bibinfo}[2]{#2}
\providecommand{\BIBentrySTDinterwordspacing}{\spaceskip=0pt\relax}
\providecommand{\BIBentryALTinterwordstretchfactor}{4}
\providecommand{\BIBentryALTinterwordspacing}{\spaceskip=\fontdimen2\font plus
\BIBentryALTinterwordstretchfactor\fontdimen3\font minus
  \fontdimen4\font\relax}
\providecommand{\BIBforeignlanguage}[2]{{%
\expandafter\ifx\csname l@#1\endcsname\relax
\typeout{** WARNING: IEEEtran.bst: No hyphenation pattern has been}%
\typeout{** loaded for the language `#1'. Using the pattern for}%
\typeout{** the default language instead.}%
\else
\language=\csname l@#1\endcsname
\fi
#2}}
\providecommand{\BIBdecl}{\relax}
\BIBdecl

\bibitem{basar_wireless_2019}
E.~Basar, M.~Di~Renzo, J.~De~Rosny, M.~Debbah, M.-S. Alouini, and R.~Zhang,
  ``Wireless communications through reconfigurable intelligent surfaces,''
  \emph{IEEE Access}, vol.~7, pp. 116\,753--116\,773, 2019.

\bibitem{basar_reconfigurable_2020}
E.~Basar, ``{Reconfigurable intelligent surface-based index modulation: {A} new
  beyond {MIMO} paradigm for {6G}},'' \emph{{IEEE} Trans. Commun.}, vol.~68,
  no.~5, pp. 3187--3196, May 2020.

\bibitem{elmossallamy_reconfigurable_2020}
M.~A. ElMossallamy, H.~Zhang, L.~Song, K.~G. Seddik, Z.~Han, and G.~Y. Li,
  ``Reconfigurable intelligent surfaces for wireless communications:
  Principles, challenges, and opportunities,'' \emph{IEEE Trans. on Cognitive
  Commun. and Netw.}, vol.~6, no.~3, pp. 990--1002, Sep. 2020.

\bibitem{huang_holographic_2020}
C.~Huang, S.~Hu, G.~C. Alexandropoulos, A.~Zappone, C.~Yuen, R.~Zhang,
  M.~Di~Renzo, and M.~Debbah, ``Holographic {MIMO} surfaces for {6G} wireless
  networks: Opportunities, challenges, and trends,'' \emph{IEEE Wireless
  Commun., to be published}.

\bibitem{huang_reconfigurable_2019}
C.~Huang, A.~Zappone, G.~C. Alexandropoulos, M.~Debbah, and C.~Yuen,
  ``Reconfigurable intelligent surfaces for energy efficiency in wireless
  communication,'' \emph{{IEEE} Trans. Wireless Commun.}, vol.~18, no.~8, pp.
  4157--4170, Aug. 2019.

\bibitem{wu_joint_TWC_2019}
Q.~Wu and R.~Zhang, ``Intelligent reflecting surface enhanced wireless network
  via joint active and passive beamforming,'' \emph{{IEEE} Trans. Wireless
  Commun.}, vol.~18, no.~11, pp. 5394--5409, Nov. 2019.

\bibitem{nadeem_intelligent_2019}
\BIBentryALTinterwordspacing
Q.-U.-A. Nadeem, A.~Kammoun, A.~Chaaban, M.~Debbah, and M.-S. Alouini,
  ``Intelligent reflecting surface assisted wireless communication: Modeling
  and channel estimation,'' \emph{arXiv:1906.02360 [cs, eess, math]}. [Online].
  Available: \url{http://arxiv.org/abs/1906.02360}
\BIBentrySTDinterwordspacing

\bibitem{guo_weighted_2019}
H.~Guo, Y.-C. Liang, J.~Chen, and E.~G. Larsson, ``Weighted sum-rate
  maximization for reconfigurable intelligent surface aided wireless
  networks,'' \emph{{IEEE} Trans. Wireless Commun.}, vol.~19, no.~5, pp.
  3064--3076, May 2020.

\bibitem{huang_reconfigurable_2020}
C.~Huang, R.~Mo, and C.~Yuen, ``Reconfigurable intelligent surface assisted
  multiuser {MISO} systems exploiting deep reinforcement learning,''
  \emph{{IEEE} J. Sel. Areas Commun.}, vol.~38, no.~8, pp. 1839--1850, Aug.
  2020.

\bibitem{yang_irs_meets_ofdm_2019}
Y.~Yang, B.~Zheng, S.~Zhang, and R.~Zhang, ``Intelligent reflecting surface
  meets {OFDM}: {Protocol} design and rate maximization,'' \emph{{IEEE} Trans.
  Commun.}, vol.~68, no.~7, pp. 4522--4535, Jul. 2020.

\bibitem{taha_enabling_2019}
A.~Taha, M.~Alrabeiah, and A.~Alkhateeb, ``Deep learning for large intelligent
  surfaces in millimeter wave and massive {MIMO} systems,'' in \emph{{IEEE}
  {Global} {Commun.} {Conf.} ({GLOBECOM})}, Waikoloa, HI, Dec. 2019.

\bibitem{hu_reconfigurable_2019}
J.~Hu, H.~Zhang, B.~Di, L.~Li, L.~Song, Y.~Li, Z.~Han, and H.~V. Poor,
  ``Reconfigurable intelligent surfaces based {RF} sensing: {Design},
  optimization, and implementation,'' \emph{{{IEEE} J. Sel. Areas Commun., to
  be published}}.

\bibitem{di_hybrid_2020}
B.~Di, H.~Zhang, L.~Song, Y.~Li, Z.~Han, and H.~V. Poor, ``Hybrid beamforming
  for reconfigurable intelligent surface based multi-user communications:
  Achievable rates with limited discrete phase shifts,'' \emph{{IEEE} J. Sel.
  Areas Commun.}, vol.~38, no.~8, pp. 1809--1822, Aug. 2020.

\bibitem{nadeem_asymptotic_2020}
Q.-U.-A. Nadeem, A.~Kammoun, A.~Chaaban, M.~Debbah, and M.-S. Alouini,
  ``Asymptotic max-min {SINR} analysis of reconfigurable intelligent surface
  assisted {MISO} systems,'' \emph{{{IEEE} Trans. Wireless Commun., to be
  published}}.

\bibitem{schelstraete_4KQAMfeasbility_2019}
S.~Schelstraete, D.~Dash, and I.~Latif, ``Feasbility of 4096-{QAM},''
  \emph{IEEE 802.11-19/0637}, May 2019.

\bibitem{roy_effective_2007}
O.~Roy and M.~Vetterli, ``The effective rank: {A} measure of effective
  dimensionality,'' in \emph{{European} {Signal} {Processing} {Conf.}},
  Pozna\`{n}, Poland, Sep. 2007, pp. 606--610.

\bibitem{garey_computers_1979}
M.~R. Garey and D.~S. Johnson, \emph{Computers and Intractability: {A} Guide to
  the Theory of {NP}-Completeness}.\hskip 1em plus 0.5em minus 0.4em\relax USA:
  W. H. Freeman \& Co., 1979.

\bibitem{hougne_optimally_2019}
P.~d. Hougne, M.~Fink, and G.~Lerosey, ``\BIBforeignlanguage{en}{Optimally
  diverse communication channels in disordered environments with tuned
  randomness},'' \emph{\BIBforeignlanguage{en}{Nature Electronics}}, vol.~2,
  no.~1, pp. 36--41, Jan. 2019.

\bibitem{heath_antenna_2001}
R.~Heath, S.~Sandhu, and A.~Paulraj, ``Antenna selection for spatial
  multiplexing systems with linear receivers,'' \emph{{IEEE} Commun. Lett.},
  vol.~5, no.~4, pp. 142--144, Apr. 2001.

\bibitem{nocedal_numerical_2006}
J.~Nocedal and S.~J. Wright, \emph{\BIBforeignlanguage{en}{Numerical
  optimization}}, 2nd~ed.\hskip 1em plus 0.5em minus 0.4em\relax New York:
  Springer, 2006.

\end{thebibliography}

\end{document}